\documentclass[lettersize,journal]{IEEEtran}
\usepackage{mathtools,amsmath,amsfonts,amssymb,amsthm}
\usepackage{algorithmic}
\usepackage{algorithm}
\usepackage{array}
\usepackage[caption=false,font=normalsize,labelfont=sf,textfont=sf]{subfig}
\usepackage{textcomp}
\usepackage{stfloats}
\usepackage{url}
\usepackage{verbatim}
\usepackage{graphicx}
\graphicspath{{Figures/}}
\usepackage{cite}
\usepackage{multirow}
\usepackage{acronym}
\usepackage{lipsum}
\usepackage{nccmath}

\usepackage{cuted}
\newtheorem{lemma}{Lemma}

\begin{document}

%
\title{An Unsupervised Machine Learning Scheme for Index-Based CSI Feedback in Wi-Fi
\thanks{This research is supported by the National Science Foundation (NSF) through the award number CNS 1814727.

Mrugen Deshmukh was with InterDigital, Inc., New York, NY 10120 USA. He is now with the Department of Electrical and Computer Engineering, North Carolina State University, Raleigh, NC 27606 USA (email: madeshmu@ncsu.edu).

Zinan Lin, Hanqing Lou, and Rui Yang are with InterDigital, Inc., New York, NY 10120 USA (email: Zinan.Lin@InterDigital.com, Hanqing.Lou@InterDigital.com, Rui.Yang@InterDigital.com).

Mahmoud Kamel is with InterDigital, Inc., Montreal, QC H3A 3G4 Canada (email: Mahmoud.Kamel@InterDigital.com).

\.{I}smail G{\"{u}}ven{\c{c}} is with the Department of Electrical and Computer Engineering, North Carolina State University, Raleigh, NC 27606 USA (email: iguvenc@ncsu.edu).}}

\author{
	\IEEEauthorblockN{Mrugen Deshmukh, Zinan Lin, Hanqing Lou, Mahmoud Kamel, Rui Yang and \.{I}smail G{\"{u}}ven{\c{c}}}	
}

\maketitle

\begin{abstract}
With the ever-increasing demand for high-speed wireless data transmission, beamforming techniques have been proven to be crucial in improving the data rate and the signal-to-noise ratio (SNR) at the receiver. However, they require feedback mechanisms that need an overhead of information and increase the system complexity, potentially challenging the efficiency and capacity of modern wireless networks. This paper investigates novel index-based feedback mechanisms that aim at reducing the beamforming feedback overhead in Wi-Fi links. The proposed methods mitigate the overhead by generating a set of candidate beamforming vectors using an unsupervised learning-based framework. The amount of feedback information required is thus reduced by using the index of the candidate as feedback instead of transmitting the entire beamforming matrix. We explore several methods that consider different representations of the data in the candidate set. In particular, we propose five different ways to generate and represent the candidate sets that consider the covariance matrices of the channel, serialize the feedback matrix, and account for the effective distance, among others. Additionally, we also discuss the implications of using partial information in the compressed beamforming feedback on the link performance and compare it with the newly proposed index-based methods. Extensive IEEE 802.11 standard-compliant simulation results show that the proposed methods effectively minimize the feedback overhead, enhancing the throughput while maintaining an adequate link performance. 
\end{abstract}


\section{Introduction} \label{sec:intro}

Wi-Fi technology has become an essential part of our daily lives, providing wireless Internet connectivity for a wide range of devices and applications and enabling seamless Internet access. According to a recent report by the Wi-Fi Alliance \cite{telecom2018}, the economic value provided by Wi-Fi reached \$3.3 trillion in 2021 and is expected to reach \$4.9 trillion by 2025. The number of Wi-Fi hotspots is anticipated to grow four times from approximately 169 million in 2018 to approximately 628 million in 2023 \cite{cisco_report}. As the reliance on Wi-Fi networks continues to grow, the demand for higher data rates, improved reliability, and enhanced network capacity has become increasingly critical. Future Wi-Fi generations are expected to keep meeting these demands and deliver a superior experience. 

Beamforming technology has emerged as a key enabler for achieving improved link performance and spectral efficiency in Wi-Fi systems. By utilizing multiple antennas at the transmitter, beamforming technology allows for an improved signal-to-noise ratio (SNR) at the receiver. This significantly improves the error-rate performance of a link. To fully utilize the benefits of beamforming, acquiring accurate channel state information (CSI) at the transmitter is critical. CSI provides knowledge about the characteristics of the wireless channel that may be affected by fading, multi-path, path loss, and interference. By using the CSI feedback, the transmitter can adapt the beamforming weights and optimize the reception of signals, improving the link quality and subsequently the data rate performance. Traditional channel sounding procedures used to obtain CSI involve transmitting pilot symbols known both to the transmitter as well as the receiver and estimating the channel response at the receiver side. 

As the number of antennas continues to increase to satisfy demand for higher data rates, the traditional channel sounding procedures face several challenges. The overhead associated with acquiring accurate CSI feedback escalates with the growing antenna count leading to increased complexity and longer training times, which may hamper the throughput gain. This poses significant challenges for real-time applications that require low latency and efficient feedback mechanisms. To address the challenges posed by traditional channel sounding procedures in MIMO-based Wi-Fi systems, literature has been exploring novel approaches to reduce beamforming feedback overhead, which we will overview in Section~\ref{sec:lit_review}. One promising avenue of research that has been popular in recent literature is the application of neural network (NN) based machine learning (ML) approaches for CSI feedback reduction in both 3GPP and Wi-Fi systems. NNs have demonstrated remarkable abilities in learning complex patterns within large data sets, and by training NNs on empirical CSI data, they can effectively estimate the CSI for a given received symbol, reducing the reliance on extensive traditional channel sounding procedures \cite{8322184}. This approach has been shown to allow faster and more efficient acquisition of CSI information, contributing to a reduction in beamforming feedback overhead. 

In our previous work \cite{9860553}, we presented a novel ML-aided approach where we use an index-based approach for beamforming feedback in Wi-Fi links. In particular, the existing Wi-Fi standard uses two sets of angles, $\phi$, and $\psi$ to represent a compressed version of the beamforming matrix. We obtained a candidate set by using \textit{k}-means clustering on a large empirical data set of vectors containing quantized representations of $\phi$ and $\psi$ angles. We demonstrated how using an index-based approach can reduce the amount of information required in beamforming feedback and subsequently result in a significant gain in the goodput. 

\begin{table*}[t]
  \centering
  \caption{ML-based beamforming feedback overhead reduction in literature.}
  \label{tab:MLreview}
  \begin{tabular}{|p{0.075\linewidth}|p{0.075\linewidth}|p{0.12\linewidth}|p{0.35\linewidth}|p{0.15\linewidth}|p{0.075\linewidth}|}
    \hline
    
    \multirow{2}{*}{\textbf{Ref.}} & \multirow{2}{*}{\textbf{Standard}}  & \multirow{2}{*}{\textbf{ML Architecture}}& \multirow{2}{*}{\textbf{Feedback Approach}} & \multirow{2}{*}{\textbf{Metric}} & \multirow{2}{*}{\textbf{Complexity}}\\
    &  &  &  &  & \\
    \hline
    This work & WLAN  & K-means & Index of the closest candidate. Effect of the representation of data in each candidate discussed. & PER, Goodput & Low \\
    \hline
    \cite{9860553} & WLAN  & K-means & Index of the closest candidate. Each candidate represented using quantized $\phi$ and $\psi$ angle indexes. & PER, Goodput & Low \\
    \hline    
    \cite{samsung_csi} & WLAN & K-means & Two steps of feedback, one for the entire sub-band and another for each subcarrier using a combination of a codebook and Givens rotation representation. & PER, Throughput & Medium \\ 
    \hline
    \cite{9448323} & WLAN & DNN & DNN based compression operating jointly with a DNN based resource allocation algorithm. & EVM, Throughput & High \\  
    \hline
    \cite{9259366}  & WLAN & Autoencoder & Compress the $\phi$ and $\psi$ angles separately using autoencoders. & EVM, Throughput & High \\ 
    \hline
    \cite{9860400} & Cellular & Deep RL & Data driven precoder without relying on fixed channel probability distribution combined with an unsupervised learning based codebook.  & BER & High \\ 
    \hline
    \cite{8322184} & Cellular & Convolutional NN & Seminal paper that introduced NN based CSI compression for cellular networks (CsiNet). & NMSE, GCS & High \\ 
    \hline
    \cite{9373670} & Cellular & Binary NN & Using BNNs to improve performance and save memory compared to CsiNet \cite{8322184}. & NMSE, Validation loss & High \\ 
    \hline
    \cite{9797871} & Cellular & Convolutional NN & Dilated convolution based encoders and decoders to compress CSI. & Complexity, NMSE & High \\ 
    \hline
    \cite{8482358} & Cellular & LSTM & Use LSTM combined with CsiNet \cite{8322184} to improve CSI recovery. & NMSE, GCS, Run time & High \\ 
    \hline
    \cite{9768327} & Cellular & Convolutional NN & Compress CSI by removing informative redundancy in terms of the mutual information. & NMSE, BER & High \\   
    \hline
    \cite{9279228} & Cellular & Autoencoder & Compression with the aim of improving the beamforming performance gain instead of improving the CSI feedback accuracy. & Data rate & High \\ 
    \hline
    \cite{9662381} & Cellular & Autoencoder & Use NN based architecture to replace implicit feedback methods and improve link-level performance. & GCS, Complexity & High \\ 
    \hline
    \cite{9705497} & Cellular & Transformer & Use transformer architecture to improve the quality of the compressed CSI. & NMSE, Complexity & High \\ 
    \hline
    \cite{8885897} & Cellular & LSTM, Transformer & Use LSTM architecture in the CSI encoding procedure and transformer in the decoding procedure to improve accuracy of reported CSI. & NMSE, GCS & High \\ 
    \hline
  \end{tabular}
\end{table*}

In this paper, which significantly extends the contribution in \cite{9860553}, we explore several different methodologies to generate and represent the candidate sets. In particular, we propose five different approaches and their respective methodologies to generate candidate sets, which are based on: 1) Lower quantization for $\psi$ angles in compressed beamforming, 2) Separate clustering of $\phi$ and $\psi$ angles in compressed beamforming, 3) Serialized beamforming steering matrices, 4) Covariance matrices of the corresponding channel matrices and 5) Accounting for \textit{effective} distance between the angles in iFOR \cite{9860553}. We demonstrate how the actual representation of the empirical data set may affect the clustering and generation of the candidates. We then perform an extensive performance comparison of these methodologies using simulation results and discuss the related trade-offs, where the highest improvement in goodput is shown to be approximately $54\%$ at high SNR. We also discuss the utility of transmitting partial compressed beamforming feedback, specifically in cases of multiple spatial streams being utilized to transmit data. We demonstrate that using this relatively simple approach may lead to approximately $17\%$ gain in the goodput at high SNRs.  
    
The rest of the paper is organized as follows. Section~\ref{sec:lit_review} covers the literature survey and discusses the related work. In Section~\ref{sec:sys_model}, we describe the system model that we use in our simulations and performance evaluation. In Section~\ref{sec:index_fb}, we describe the proposed methods that use index-based feedback, while Section~\ref{sec:fixed_psi} discusses an additional method to reduce the beamforming feedback using only partial information from the beamforming method currently used in the 802.11be standard. In Section~\ref{sec:sim_results} we present the performance evaluation of the proposed methods and the corresponding trade-offs. In Section~\ref{sec:conclusion} we conclude the paper and discuss the possible future work.

\section{Literature Review} \label{sec:lit_review}

Our literature review is directed toward the works that discuss channel sounding and beamforming overhead reduction in the Wi-Fi and cellular networks. Initially, we delve into some conventional methodologies, which provide a fundamental understanding and context of the overhead reduction problem. Following this, we turn our focus toward more contemporary approaches that use ML-based tools. Specifically, we concentrate on the recent developments in using NN-based architectures to address the CSI overhead reduction problem.

\subsection{Traditional Approaches}

The problem of overhead reduction has been in discussion in Wi-Fi since the induction of MIMO. Traditionally, several works have addressed this issue by optimizing channel sounding parameters. Works such as \cite{7442585, 6990336, 8013252, nabetani2018novel, 8516958} explore methodologies where the frequency of channel sounding is reduced assuming the quasi-static nature of the Wi-Fi channels. In \cite{6666172}, reducing overhead for multi-user MIMO (MU-MIMO) systems is discussed considering implicit as well as explicit feedback methods, whereas \cite{7809634} focuses on implicit feedback methods. Works in \cite{6328529, 8464287, 7354785} study the relationship between the different channel sounding parameters and the system throughput for MU-MIMO. 

\subsection{Machine Learning Based Approaches}

There has been a surge of interest from the research community and industry to explore the use of ML-based tools to solve problems in the wireless communications domain \cite{9786784}. In Table~\ref{tab:MLreview}, we review some of the papers in the recent literature that use ML-based solutions to address the CSI feedback reduction problem in wireless networks. For each reference listed in Table~\ref{tab:MLreview}, we discuss the wireless standard the paper focuses on, the basic ML architecture being used, the performance metrics being used to evaluate the proposed methods, and the complexity of implementation of the proposed methods in a practical wireless network.

The abbreviations in Table~\ref{tab:MLreview} are as follows. DNN stands for deep neural network, RL is reinforcement learning, LSTM is the long short-term memory architecture, EVM is error magnitude vector at the receiver, BER is the bit error rate, and the generalized cosine similarity (GCS) and the normalized mean squared error (NMSE) are the intermediate key performance indicators (KPIs) that are used to measure the accuracy of the beamforming feedback generated by a proposed scheme. 

ML-based solutions for CSI feedback reduction have been extensively studied in the literature for cellular networks and remain an interesting topic for the community as the field of ML evolves. We refer the interested readers to the relevant surveys, e.g., \cite{9931713} and \cite{9896861}.

\begin{figure}[t] 
    \centering\vspace{-3mm}
    \includegraphics[width=0.99\linewidth,trim={1cm 3cm 1cm 2.5cm},clip]{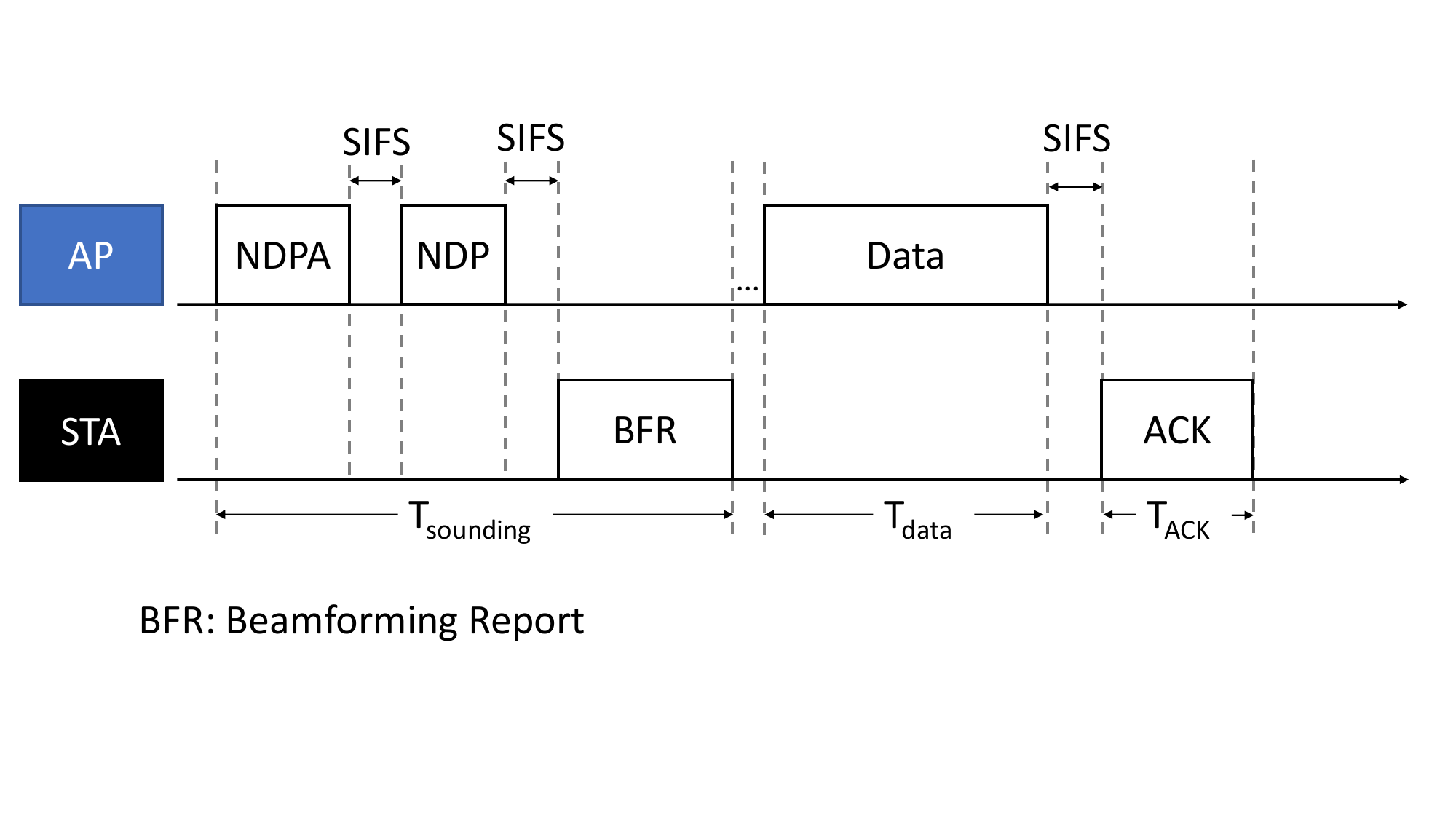}\vspace{-4mm}
    \caption{SU-MIMO downlink transmission procedure.}
    \label{fig:channel_sounding}\vspace{-4mm}
\end{figure}

\section{System Model} \label{sec:sys_model}

In this work, we consider a single-user MIMO (SU-MIMO) Wi-Fi link operating in downlink using transmit beamforming (TxBF). TxBF enables improvement in the link performance using weights that are applied to the transmitted signal at each antenna to improve the SNR at the receiver. These so-called ``weights" are adopted from the knowledge of the propagation environment, i.e., the CSI. 

\subsection{SU-MIMO Downlink Transmission in Wi-Fi}

Mathematical representation of a general MIMO system using TxBF can be given by \cite{stacey_book}
\begin{equation}\label{eq:downlink_Y}
    \mathbf{y} = \sqrt{\frac{\rho}{N_\text{r}}}\mathbf{H}\hat{\mathbf{V}}\mathbf{x} + \mathbf{z},
\end{equation}
where, the received signal $\mathbf{y}$ is a column vector of size $N_\text{t}\times 1$, the transmitted signal $\mathbf{x}$ is a column vector of size $N_\text{c}\times 1$, $\mathbf{H}$ is the MIMO channel matrix of size $N_\text{t}\times N_\text{r}$ and the AWGN $\mathbf{z}$ is a vector of size $N_\text{t}\times 1$, and $\rho$ is the SNR. The dimensions of these matrices are as follows: $N_\text{r}$ is the number of antennas at the beamformer, $N_\text{t}$ is the number of antennas at the beamformee, and $N_\text{c}$ is the number of spatial streams being transmitted. Here, we assume $N_\text{c} = N_\text{t}$.

The matrix $\hat{\mathbf{V}}$ is referred to as the steering matrix or the precoding matrix and is obtained from the singular value decomposition of $\mathbf{H}$. For efficient utilization of TxBF, the weights in the matrix $\hat{\mathbf{V}}$ need to be determined to the highest accuracy possible before data transmission begins. 

In Wi-Fi, a channel sounding mechanism is typically used to determine the CSI, from which the matrix $\hat{\mathbf{V}}$ is obtained. Fig.~\ref{fig:channel_sounding} describes an example SU-MIMO downlink transmission process in Wi-Fi, including the channel sounding mechanism, between an access point (AP), the beamformer, and a non-AP station (STA), the beamformee.

\begin{figure}[t] 
    \centering\vspace{-3mm}
    \includegraphics[width=0.99\linewidth,trim={3.25cm 7cm 3.25cm 7cm},clip]{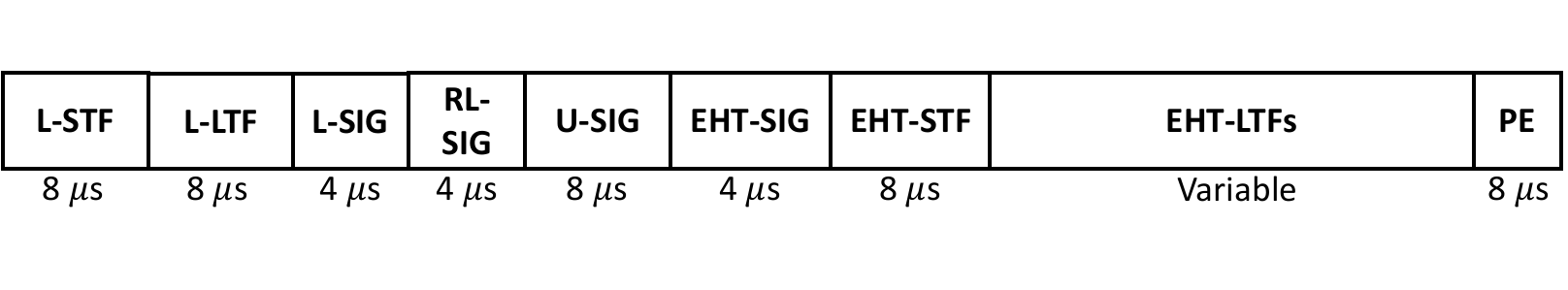}\vspace{-4mm}
    \caption{Sounding NDP format for 802.11be.}
    \label{fig:ppdu_format}\vspace{-4mm}
\end{figure}

The entire downlink transmission process can be broadly divided into the following three phases: 
\begin{itemize}
	\item Channel sounding: The channel sounding phase begins with the AP transmitting a null data packet announcement (NDPA) frame, which contains relevant information such as the address of the beamformee, sounding channel bandwidth, feedback type, etc. After a short inter-frame space (SIFS) interval, the AP transmits a null data packet (NDP). This NDP follows the same format of a multi-user physical layer protocol data unit (MU-PPDU), except the data field is excluded. The NDP format for Wi-Fi 7, or IEEE 802.11be is shown in Fig.~\ref{fig:ppdu_format} \cite{11bedraft}. The non-AP STA then uses the extremely high throughput (EHT) long training fields (LTFs) in this PPDU to estimate the CSI and determine $\mathbf{V}$. After another SIFS interval, the STA transmits the information pertaining to $\mathbf{V}$ in the beamforming report (BFR) in the uplink.
	\item Data transmission: The AP can now use the BFR to reconstruct the steering matrix $\hat{\mathbf{V}}$. This $\hat{\mathbf{V}}$ matrix is then used to transmit the data in downlink as shown in (\ref{eq:downlink_Y}). In most practical scenarios, the reconstructed $\hat{\mathbf{V}}$ is not the same as the $\mathbf{V}$ calculated at the STA during channel sounding. This loss in accuracy may lead to performance degradation of the link. 
	\item ACK transmission: After a SIFS of the Data frame, the STA sends an ACK frame if the transmitted data frame is successfully received.  
\end{itemize}

There are several ways to represent the $\mathbf{V}$ matrix information, which will be briefly discussed in the next subsection. 

\subsection{Compressed Beamforming Feedback in Wi-Fi}
The beamforming feedback in Wi-Fi systems can be broadly categorized into two types: implicit feedback and explicit feedback \cite{stacey_book}. Subsequently, explicit feedback can be further divided into three types: CSI feedback, non-compressed beamforming weights feedback, and compressed beamforming weights feedback. In this work, we focus on the compressed beamforming type of feedback, which is typically used in Wi-Fi as it enables to reduce the feedback overhead. 

During the channel sounding procedure, the beamformee computes $\mathbf{V}$ from the Singular Value Decomposition (SVD) of the estimated channel matrix ($\mathbf{H}$). This procedure needs to be repeated for every active subcarrier and the corresponding data is transmitted to the beamformer in the BFR in uplink. Feeding back quantized values of all the complex elements in $\mathbf{V}$ for every active subcarrier, however, may lead to significant overhead in the BFR. In the compressed beamforming feedback method, the beamformee will compress and approximate the $\mathbf{V}$ matrix in the form of a function of angles \cite{stacey_book}.

The compressed representation of $\mathbf{V}$ can be mathematically expressed as
\begin{equation}
    \mathbf{V} = \Bigg(\prod_{{i}=1}^{\text{min}(N_\text{c}, N_\text{r}-1)} \Bigg[ \mathbf{D}_{i}\prod_{{l}={i}+1}^{N_r} \mathbf{G}^T_{li}(\psi_{l,i})\Bigg] \times \Tilde{ \mathbf{I}}_{N_\text{r}\times N_\text{c}}\Bigg),
\end{equation}
where $\Tilde{ \mathbf{I}}_{N_\text{r}\times N_\text{c}}$is an identity matrix padded with $0$s to fill the additional rows or columns when $N_\text{r}\neq N_\text{c}$. Unless mentioned otherwise, we assume $N_\text{c}=N_\text{t}$. The matrix $\mathbf{V}$ is of the dimensions $N_\text{r}\times N_\text{c}$ \cite{stacey_book}. In this paper, we consider a non-trigger-based (TB) channel sounding sequence during which $N_\text{r}$ and $N_\text{c}$ are determined by the beamformee. 

The matrix $\mathbf{G}_{li} (\psi)$ is a Givens rotation matrix of dimensions $N_\text{r}\times N_\text{r}$ as shown below
\begin{equation}\small
    \mathbf{G}_{li}(\psi) = \begin{bmatrix} 
        \mathbf{I}_{i-1} & 0 & 0 & 0 & 0 \\ 
        0 & \cos(\psi) & 0 & \sin(\psi) & 0 \\ 
        0 & 0 & \mathbf{I}_{l-i-1} & 0 & 0 \\ 
        0 & -\sin(\psi) & 0 & \cos(\psi) & 0 \\ 
        0 & 0 & 0 & 0 & \mathbf{I}_{N_\text{r}-l} 
        \end{bmatrix},
\label{eq:g_psi_matrix}
\end{equation}
and the matrix $\mathbf{D}_i$ is a diagonal matrix also of the dimensions $N_\text{r}\times N_\text{r}$ and is represented as 
\begin{equation}
    \mathbf{D}_i = \begin{bmatrix} 
        \mathbf{I}_{i-1} & 0 & \hdots & \hdots & 0 \\ 
        0 & e^{j\phi_{i,i}} & 0 & \hdots & 0 \\ 
        \vdots & 0 & \ddots & 0 & \vdots \\ 
        \vdots & \hdots & 0 & e^{j\phi_{N_\text{r}-1, i}} & 0 \\ 
        0 & \hdots & \hdots & 0 & 1
        \end{bmatrix}
\label{eq:d_phi_matrix}
\end{equation} 
where each $\mathbf{I}_\text{m}$ is an $\text{m}\times \text{m}$ identity matrix. When the beamformee is requested to send the BFR, it will actually report a vector containing the indices of quantized values of the $\phi$ and $\psi$ angles, which are used to reconstruct the $\mathbf{V}$ matrix by the beamformer. 

\begin{table}[t]
  \centering
  \caption{Order of angles in the BFR using compressed beamforming}
  \label{tab:fb_vec_examples}
  \begin{tabular}{|p{0.12\linewidth}|p{0.1\linewidth}|p{0.6\linewidth}|}
    \hline
    $N_\text{r}\times N_\text{c}$ & \textbf{No. of angles} & \textbf{Angle report vector} \\
    \hline
    $2\times 1$ & $2$ & $\phi_{1,1}$, $\psi_{2,1}$ \\
    \hline
    $2\times 2$ & $2$ & $\phi_{1,1}$, $\psi_{2,1}$ \\
    \hline
    $4\times 1$ & $6$ & $\phi_{1,1}$, $\phi_{2,1}$, $\phi_{3,1}$, $\psi_{2,1}$, $\psi_{3,1}$, $\psi_{4,1}$ \\
    \hline
    $4\times 2$ & $10$ & $\phi_{1,1}$, $\phi_{2,1}$, $\phi_{3,1}$, $\psi_{2,1}$, $\psi_{3,1}$, $\psi_{4,1}$, $\phi_{2,2}$, $\phi_{3,2}$, $\psi_{3,2}$, $\psi_{4,2}$ \\
    \hline
  \end{tabular}
\end{table}

Thus, the $\mathbf{V}$ matrix may be completely represented using the $\mathbf{\Phi}=\{\phi_{l,i}\}, l\in\{i+1, ... , N_r\}$ and the $\mathbf{\Psi}=\{\psi_{i,i}\}, i\in\{1, ... , \text{min}(N_\text{c},N_\text{r}-1)\}$ angles. When generating the BFR, the beamformee will generate a \textit{vector of indices} of the quantized values of the $\mathbf{\Phi}$ and  $\mathbf{\Psi}$ angles. The angles in $\mathbf{\Phi}$ can be quantized as 
\begin{equation}\label{eq:phi_bf}
    \phi_{l,i}^q = \pi\Bigg( \frac{1}{2^{b_{\phi}}} + \frac{q}{2^{b_{\phi}-1}} \Bigg),~ q=0,1, ... , 2^{b_{\phi}}-1,
\end{equation}
and the angles in $\mathbf{\Psi}$ can be quantized as
\begin{equation}\label{eq:psi_bf}
    \psi_{i,i}^q = \pi\Bigg( \frac{1}{2^{b_{\psi}+2}} + \frac{q}{2^{b_{\psi}+1}} \Bigg),~ q=0,1, ... , 2^{b_{\psi}}-1,
\end{equation}
where $b_{\phi}$ and $b_{\psi}$ are the number of bits used to represent the $\mathbf{\Phi}$ and $\mathbf{\Psi}$ angles respectively. In 802.11be, $b_{\phi}$ and $b_{\psi}$ are indicated by the NDP Announcement frame sent from the beamformer. The length of the feedback vector with quantized indices depends on $N_\text{c}$ and $N_\text{r}$. Table~\ref{tab:fb_vec_examples} shows examples of such vectors for different MIMO settings.

Using compressed beamforming, an approximate form of the $\mathbf{V}$ matrix is generated. The accuracy of this approximation depends on the quantization order chosen in (\ref{eq:phi_bf}) and (\ref{eq:psi_bf}), with a higher quantization level leading to higher accuracy. Hence $b_{\phi}$ and $b_{\psi}$ are important parameters, as the loss in accuracy of $\mathbf{V}$ will subsequently affect the link performance. Using an index-based approach to beamforming feedback will lead to similar challenges. The higher the size of the candidate set, the higher the expected accuracy of the steering matrix used in data transmission leading to better link performance. 

\section{Index-Based Beamforming Feedback Generation} \label{sec:index_fb}

In this section, we describe the proposed index-based methods. The flow of the index-based beamforming feedback is expressed in Fig.~\ref{fig:candidate_selection}. As described in Section~\ref{sec:sys_model}, the beamformee (STA) uses the LTF symbols in the NDP to estimate the channel ($\mathbf{H}$). Subsequently, the steering matrix ($\mathbf{V}$) is then determined from the SVD of $\mathbf{H}$. In the compressed beamforming representation, this matrix $\mathbf{V}$ is then represented in terms of angles and the indices of the quantized values of these angles are sent in the feedback.

In the index-based methods, however, after the $\mathbf{V}$ matrix is computed it is compared to a set of pre-determined feedback matrices, which we refer to as the candidate set. This candidate set is assumed to be available at both the beamformer and the beamformee. Then, using a pre-defined distance metric, e.g. squared Euclidean distance (SED), the closest candidate to the $\mathbf{V}$ matrix is determined. The index of this closest candidate is then sent in the beamforming feedback instead of sending the complete information about the $\mathbf{V}$ matrix. Since the beamformer (AP) also has the same candidate set, it can use the received index from the STA to pick the candidate steering matrix ($\hat{\mathbf{V}}$). The AP then uses this $\hat{\mathbf{V}}$ matrix to transmit the data in the downlink.

\begin{figure}[t] 
    \centering\vspace{-3mm}
    \includegraphics[width=0.99\linewidth,trim={5.5cm 3.5cm 5cm 4cm},clip]{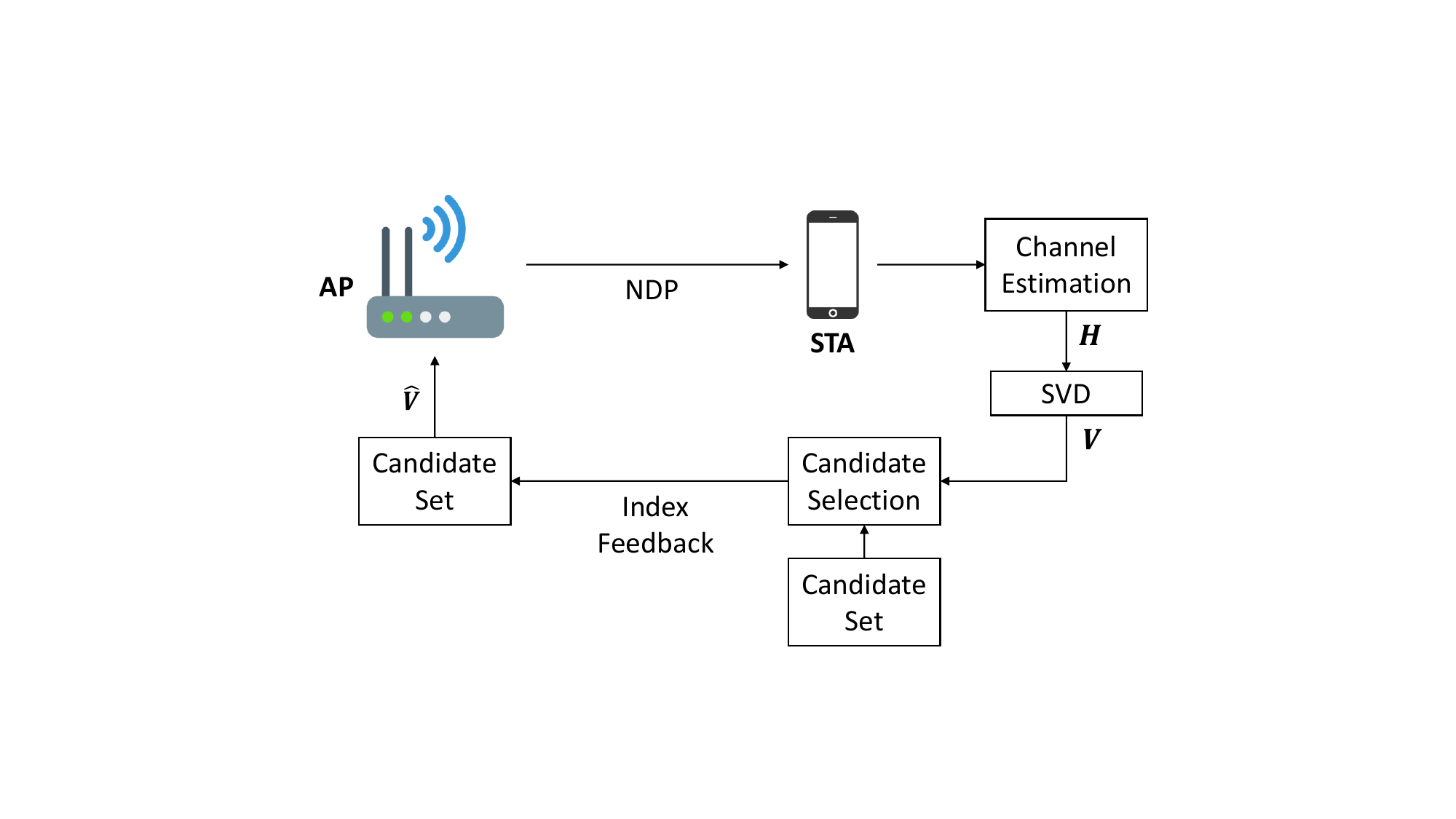}\vspace{-4mm}
    \caption{Channel sounding feedback using a candidate set.}
    \label{fig:candidate_selection}\vspace{-4mm}
\end{figure}

The candidate sets used in the proposed index-based methods are generated using a clustering algorithm. More specifically, we generate a large empirical data set by simulating the beamforming feedback over a set of diverse channel models. This data set is then fed into a \textit{k}-means clustering algorithm \cite{kmeans_arthur}. The \textit{k}-means clustering algorithm enables us to choose the number of clusters to generate for a given data set. The centroid of these clusters may be used as the candidates in our proposed methods. In the following subsections, we describe our proposed index-based beamforming methods in detail. The general methodology for candidate generation and beamforming feedback is similar as in Fig.~\ref{fig:candidate_selection}, whereas the representation of the data in the data set (and the candidates) is different for each method. 

It is clear that $\mathbf{V}$ is the optimal solution for beamforming in terms of improving the SNR at the receiver. Any other solutions, such as $\hat{\mathbf{V}}$ from a candidate set could lead to inaccuracy in the BFR and subsequently PER performance degradation. In our previous work \cite{9860553}, we generated the candidate set in the form vectors of quantized angle indexes (similar to Table~\ref{tab:fb_vec_examples}). In \cite{9860553}, we observe that even though using such a candidate set results in a loss in accuracy of the $\mathbf{V}$ matrix and degrades the PER performance, the reduction in BFR is large enough to provide significant gain in the goodput. In this work, we explore different ways of generating the data sets and subsequently, the candidate set so that for the same size of the candidate set, the error-rate performance may be improved. The representation of data and the methodology for clustering using \textit{k}-means clustering for each method will be detailed in the corresponding subsections. The general framework of performing \textit{k}-means clustering on a given data set is described in Algorithm~\ref{alg:kmeans}, where $d_{\text{f}}$ is a distance metric that is specific to the given proposed method. Mathematically, it may be represented as 
\begin{equation} \label{eq:d_func}
    d_{\text{f}} = f(\mathbf{x}; \mathbf{x}_{\text{ref}}),
\end{equation}
where $\mathbf{x}$ is a data point from the given data set, $\mathbf{x}_{\text{ref}}$ represents a centroid in \textit{k}-means clustering, $f:C^n\times \{\mathbf{x}_{\text{ref}}\} \Rightarrow R^{+}$ for $X_{\text{ref}} \in C^{n}$. Here the function $f$ corresponds to the distance metric chosen for different methods that will be discussed in the following subsections.

The centroids of the clusters obtained after convergence of the \textit{k}-means clustering is considered to be the candidate set. The representation of data for each method we consider and the distance used in \textit{k}-means clustering are described in Table~\ref{tab:index_based_bf}.

\begin{table}[t]
  \centering
  \caption{Index-based methods.}
  \label{tab:index_based_bf}
  \begin{tabular}{|p{0.02\linewidth}|p{0.32\linewidth}|p{0.32\linewidth}|p{0.12\linewidth}|}
    \hline
     & \textbf{Method} & \textbf{Candidate Representation} & \textbf{Distance} ($d_{\text{f}}$) \\
    \hline
    1 & Lower Quantization Order for $\psi$ (LQP) & Angle index vector & SED \\
    \hline   
    2 & Separate Clustering of $\phi$ and $\psi$ Angles (SCP) & Angle index vector & SED \\
    \hline  
    3 & Serialized $\mathbf{V}$ (SV-SED) & Complex elements of $\mathbf{V}$ & SED \\
    \hline  
    4 & Serialized $\mathbf{V}$ (SV-CD) & Complex elements of $\mathbf{V}$ & CD \\
    \hline  
    5 & Normalized Covariance Matrix (NCM) & Complex elements of covariance matrix & SED \\
    \hline  
    6 & iFOR with effective distance (iFOR+) & Angle index vector & SED \\
    \hline  
  \end{tabular}
\end{table}

\subsection{Lower Quantization Order for $\psi$ (LQP)}
In the development of this method, the distance metric used in \textit{k}-means clustering for measuring the separation between two distinct vectors is the SED. For calculating this distance, we calculate the distance between all the indices representing the $\phi$ and the $\psi$ angles, and we sum them together to determine the total distance. The following equation shows the SED calculated ($d_{\text{SED}}$) between the $i$th vector $\mathbf{x}_i$ in the data set and the centroid of the $k$th cluster $\mathbf{c}_k$
\begin{equation} \label{eq:sed_cpp}
    d_{\text{SED}} = \lvert \mathbf{x}_i-\mathbf{c}_{k} \rvert^2 = (\mathbf{x}_i-\mathbf{c}_{k})^T(\mathbf{x}_i-\mathbf{c}_{k}),
\end{equation}
where the vectors $\mathbf{x}_i$ and $\mathbf{c}_k$ have the dimensions $N_\text{a}\times 1$, with $N_\text{a}$ being the total number of angles to be reported in the feedback vector. The centroid of a cluster in \textit{k}-means clustering is obtained by averaging over all the vectors within that particular cluster \cite{kmeans_arthur}. 

As we will discuss later in the simulation results, a loss of accuracy in the $\phi$ angles in (\ref{eq:d_phi_matrix}) is far more detrimental to the link performance (e.g. packet-error-rate) than a loss of accuracy in the $\psi$ angles in (\ref{eq:g_psi_matrix}). Noting this observation, it can be considered that having a relatively high quantization order for $\psi$ angles may add further inaccuracy in calculating the SED between two vectors. Hence, in this method of generating the candidates, we consider lower quantization orders for the $\psi$ angles. In doing so, we reduce the impact of the contribution of the $\psi$ angle indexes to the SED. Giving a finer quantization level to the $\phi$ angles in calculating the distance in this way might improve the clustering of data in favor of enhancing the link performance.

\subsection{Separate Clustering of $\phi$ and $\psi$ Angles (SCP)}
In our previous work \cite{9860553}, we feed the entire angle vectors with feedback angle indices for $\phi$ and $\psi$ angles (as also shown in Table II) to the \textit{k}-means clustering algorithm. The cluster centroids are initialized using the \textit{k}-means++ algorithm \cite{kmeans_arthur}. After the \textit{k}-means clustering converges, we obtain the cluster centroids which are the candidate vectors. 

Another approach to obtaining the candidate vectors is to separate the feedback angle indices for the $\phi$ and $\psi$ angles. Here we have two data sets, one with the $\phi$ angle index data and another with the $\psi$ angle index data. We feed these data sets to the \textit{k}-means clustering function separately and obtain the candidate sets for $\phi$ and $\psi$ angle index data correspondingly. As we will demonstrate later in the simulation results, the link performance (e.g. PER) is more sensitive to the inaccuracy in $\phi$ angles in the candidates, as opposed to the inaccuracy in the $\psi$ angles. Performing separate clustering on the two different angles would thus allow us the flexibility to assign more feedback bits (or candidate vectors) for $\phi$ angles than for $\psi$ angles.

\begin{algorithm}[t]
\caption{\textit{k}-means clustering algorithm.}\label{alg:alg1}
\begin{algorithmic}
\STATE \textbf{Input:} number of candidates $N_{\text{k}}$
\STATE \textbf{Output:} $\mathbf{c}_{\text{k}} \gets$ centroids of the $N_{\text{k}}$ clusters
\STATE \textbf{Function:} \textsc{K-MEANS}
\STATE \hspace{0.5cm} initialize cluster centroids randomly as $\mu_1, \mu_2, ..., \mu_k$
\STATE \hspace{0.5cm} \textbf{for} each data point $\mathbf{x}^i$ in the data set
\STATE \hspace{1cm} choose the closest cluster centroid using
\STATE \hspace{1cm} $c^i := \underset{j}{\mathrm{arg\,min}}~d_{\text{f}}(\mathbf{x}^i, \mu_j)$
\STATE \hspace{1cm} assign $\mathbf{x}^i$ to the cluster $c^i$
\STATE \hspace{1cm} update centroids by taking the mean of all the \\\hspace{1.1cm}data points assigned to that cluster
\STATE \hspace{0.5cm} \textbf{repeat} until convergence or specified iterations reached
\STATE \textbf{return} $\mathbf{c}_{\text{k}}$

\end{algorithmic}
\label{alg:kmeans}
\end{algorithm}

Consider an example where the number of clusters for $\phi$ vectors is $W_1$ and the number of clusters for $\psi$ vectors is $W_2$. To give more weights on $\phi$ vectors, we have $W_1>W_2$, e.g., $W_1=2^8$  and $W_2=2^2$. The total number of bits required for each subcarrier group in this example is then given by $\log_2(W_1)+\log_2(W_2)=8+2=10$ bits. The values of $W_1$ and $W_2$ may be implementation dependent and can be modified depending on the resolution required for the two angles while keeping the total required number of feedback bits the same. 

\subsection{Serialized $\mathbf{V}$ (SV)}
So far we have considered the feedback vectors with angle indices in our index-based methods. In the serialized $\mathbf{V}$ method, however, we consider the actual complex values of the unitary steering matrix $\mathbf{V}$ obtained from the SVD of the channel matrix. Consider a $M\times N$ MIMO case, where the $\mathbf{V}$ matrix can be represented as 
\begin{equation}\label{eq:v_mat_form}
    \mathbf{V} = \begin{bmatrix} 
         v_{1,1} & v_{2,1} & \hdots & v_{M,1} \\ 
         \vdots & \hdots & \hdots & \vdots \\ 
         v_{1,N} & v_{2,N} & \hdots & v_{M,N}  
        \end{bmatrix}^T,
\end{equation}
where each $v_{i,j}$ represents a complex element of the $\mathbf{V}$ matrix for the $i$th row and $j$th column.

Furthermore, we serialize this $\mathbf{V}$ matrix so that we have a single vector instead of a matrix. This vectorized form of $\mathbf{V}$ may be expressed as 
\begin{equation}
\label{eq:serialV}
    \mathbf{v}^{\text{s}} = [v_{1,1}, v_{2,1}, \hdots, v_{M,1}, \hdots, v_{1,N}, v_{2,N}, \hdots, v_{M,N}]^T.
\end{equation}

Using the vector format in (\ref{eq:serialV}), we generate a new data set of feedback vectors. This data set is then fed to the \textit{k}-means clustering algorithm. In this method, we consider two distance metrics while clustering over the data set. The first one is the squared Euclidean distance, calculated similarly as in (\ref{eq:sed_cpp}). The results using this method are labeled as \textit{SV-SED}. Another distance metric that we consider is the Cosine Distance (CD) calculated ($d^{\text{CD}}_{i,k}$) as 
\begin{equation}
    d^{\text{CD}}_{i,k} = 1 - \frac{\lvert\mathbf{c}^*_{k}\mathbf{v}^\text{s}_{i}\rvert}{\lVert\mathbf{v}^\text{s}_{i}\rVert_2 \lVert\mathbf{c}_{k}\rVert_2},
\end{equation}
where $\mathbf{v}^\text{s}_{i}$ represents the $i$th serialized vector in the data set and $\mathbf{c}_{k}$ is the $k$th cluster. Furthermore, if $\mathbf{V}$ is a unitary matrix, then $\lVert\mathbf{v}^s\rVert_2=N$ since the $L_2$ norm of each column of $\mathbf{V}$ is $1$, the above equation reduces to 
\begin{equation}
\label{eq:d_cd}
    d^{\text{CD}}_{i,k} = 1 - \frac{\lvert\mathbf{c}^*_{k}\mathbf{v}^s_{i}\rvert}{N^2}, 
\end{equation}
making the cosine distance inversely proportional to the inner product between the two vectors. The results using this method are called serialized $\mathbf{V}$ with cosine distance (\textit{SV-CD}).

After the \textit{k}-means algorithm converges and the centroids are obtained, the candidates are obtained by re-arranging the vectors into the matrix form, as in (\ref{eq:v_mat_form}). Also, to satisfy the condition of unitary matrices, the columns of these matrices are orthogonalized using the Gram-Schmidt procedure \cite{strang_algebra}.

\subsection{Normalized Covariance Matrix (NCM)}
In this method, we use an alternative representation of the beamforming feedback data instead of 1) the vector of angle indexes obtained via compressed beamforming or 2) serializing the $\mathbf{V}$ matrix. Namely, we use the covariance matrix of the channel matrix ($\mathbf{H}$) as a data point, which can be expressed as 
\begin{equation}\label{eq:cov_mat_H}
    \mathbf{K} = \mathbf{H}^H\mathbf{H},
\end{equation}
where $\mathbf{H}^H$ represents the Hermitian of the channel matrix $\mathbf{H}$. The data set of these matrices is then fed to the \textit{k}-means clustering algorithm, where the metric used to calculate the distance between a covariance matrix and the cluster centroid is the SED.

To understand the relationship between clustering using the covariance matrix and the beamforming matrix $\mathbf{V}$, let us decompose $\mathbf{H}$ using SVD as 
\begin{equation}
    \mathbf{H} = \mathbf{U}\mathbf{S}\mathbf{V}^H,
\end{equation}
where $\mathbf{U}$ and $\mathbf{V}$ are unitary matrices and $\mathbf{S}$ is a diagonal matrix with the corresponding singular values. Substituting this back into the covariance matrix in (\ref{eq:cov_mat_H}), we get 
\begin{equation}
    \mathbf{K} = (\mathbf{U}\mathbf{S}\mathbf{V}^H)^H\mathbf{U}\mathbf{S}\mathbf{V}^H = \mathbf{V}\mathbf{S}^2\mathbf{V}^H, \label{eq:cov_mat_1}
\end{equation}
as $\mathbf{U}^H\mathbf{U}=\mathbf{I}$, and $\mathbf{S}^2$ contains the square of the singular values at the diagonal.

Now consider an example case with a channel vector $\mathbf{h}$ of size $1\times M$, which would have only one singular value $s$. For this $\mathbf{h}$, the size of $\mathbf{K}$ is $M\times M$ and the size of $\mathbf{v}$ is $M\times 1$, such that $\mathbf{v}=[v_1,v_2,\hdots,v_M ]^T$. Substituting in (\ref{eq:cov_mat_1}), we have
\begin{equation}
    \mathbf{K} = s^2\begin{bmatrix} 
        v_{1}v_{1}^{*} & v_{1}v_{2}^{*} & \hdots & v_{1}v_{M}^{*}  \\ 
        v_{2}v_{1}^{*} & v_{2}v_{2}^{*} & \hdots & v_{2}v_{M}^{*}  \\ 
        \vdots & \hdots & \ddots & \vdots \\ 
        v_{M}v_{1}^{*} & v_{M}v_{2}^{*} & \hdots & v_{M}v_{M}^{*} \label{eq:norm_cov_mat}
        \end{bmatrix}.
\end{equation}

For \textit{k}-means clustering, we consider the SED between these covariance matrices. Let $\mathbf{K}_\text{a}$ and $\mathbf{K}_\text{b}$ be two different covariance matrices, each of size $M\times M$. The distance between the two matrices ($d_{\text{cov}}$) then can be computed as 
\begin{equation}\label{eq:d_cov_sed}
    d_{\text{cov}} = \sum^M_{i=1} \sum^M_{j=1} \Bigg| \frac{\mathbf{K}_\text{a}(i,j)}{\lVert \mathbf{K}_\text{a} \rVert_\text{F}} - \frac{\mathbf{K}_\text{b}(i,j)}{\lVert \mathbf{K}_\text{b} \rVert_\text{F}} \Bigg|^2,
\end{equation}
where $\lVert .\rVert_\text{F}$ represents the Frobenius norm \cite{strang_algebra} of the argument. In (\ref{eq:d_cov_sed}), each element in $\mathbf{K}_\text{a}$ and $\mathbf{K}_\text{b}$ is normalized using the Frobenius norm of the respective matrix. 

\begin{lemma}
For two covariance matrices $\mathbf{K}_\text{a}$ and $\mathbf{K}_\text{b}$ obtained from channel matrices with a single column, the squared Euclidean distance is inversely proportional to the inner product of the corresponding vectors $\mathbf{v}_\text{a}$ and $\mathbf{v}_\text{b}$:
\begin{equation} \label{eq:dcov}
    d_{\text{cov}} = 2(1-\lvert\mathbf{v}_\text{a}^H\mathbf{v}_\text{b}\rvert^2),
\end{equation}
where $\mathbf{v}_\text{a}$ and $\mathbf{v}_\text{b}$ are obtained from the SVD of $\mathbf{K}_\text{a}$ and $\mathbf{K}_\text{b}$ respectively.
\end{lemma}
\begin{proof}
See Appendix.
\end{proof}

Note that $\mathbf{v}_\text{a}$ and $\mathbf{v}_\text{b}$ in (\ref{eq:dcov}) are the unitary vectors obtained from the SVD of  $\mathbf{K}_\text{a}$ and  $\mathbf{K}_\text{b}$, respectively. Thus, for the single spatial stream case, the SED between two covariance matrices is inversely proportional to the inner product or the cosine distance between the corresponding $\mathbf{v}$ vectors. This will have implications on the performance metrics which will be discussed in the subsequent sections.


\begin{figure}[t] 
    \centering\vspace{-3mm}
    \includegraphics[width=0.95\linewidth,trim={10cm 5.5cm 10cm 5.5cm},clip]{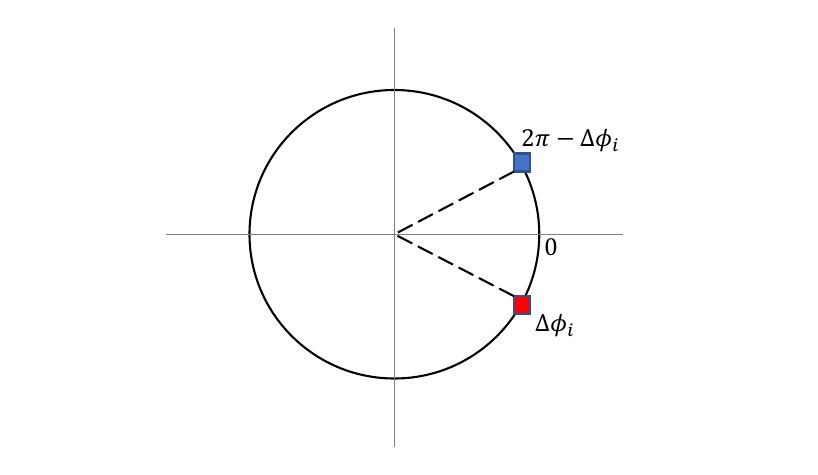}\vspace{-4mm}
    \caption{Effective distance considering difference of two $\phi$ angles.}
    \label{fig:eff_dist}\vspace{-4mm}
\end{figure}

\subsection{iFOR with Effective Distance (iFOR+)}\label{sec:ifor_plus}
In this method, we consider the effective distance (or effective difference) calculated between the $\phi~\in[0,2\pi)$ angles for candidate generation. Consider any $\phi$ angle index in a feedback vector, examples of which are shown in Table~\ref{tab:fb_vec_examples}. While clustering, the \textit{k}-means clustering algorithm calculates the squared Euclidean distance between a vector in consideration and all the cluster centroids.

Let the $i$th $\phi$ angle index in the feedback vector $\mathbf{v}_\text{a}$ be $\phi_{i,\text{a}}$ and the $i$th $\phi$ index in feedback vector $\mathbf{v}_\text{b}$ be $\phi_{i,\text{b}}$. the difference between these two angle indices is given by, 
\begin{equation}
    \Delta\phi_i = \lvert \phi_{i,\text{a}} - \phi_{i,\text{b}} \rvert,
\end{equation}
where $\Delta\phi_i\in[0,2\pi)$. Mathematically, it may be considered that the effective distance ($\Delta_{\text{eff}}$) between these two angles is, 
\begin{equation}
    	\Delta_{\text{eff}} =
	\begin{cases}
		\Delta\phi_i &\text{if } \Delta\phi_i\in[0,\pi] \\
		2\pi - \Delta\phi_i &\text{if } \Delta\phi_i\in(\pi,2\pi)
	\end{cases},
\end{equation}
where $\Delta_{\text{eff}} \in [0,\pi]$. Fig.~\ref{fig:eff_dist} illustrates using a phasor representation how $\Delta\phi_i$ and ($2\pi-\Delta\phi_i$) are the same distance away from $0$, making their effective distance be the same. 

\section{Partial Compressed Beamforming Feedback} \label{sec:fixed_psi}

An alternative approach to reduce the amount of information required in beamforming feedback is to transmit partial compressed beamforming feedback information. In the compressed beamforming methodology described in Section~\ref{sec:sys_model}, information about two types of angles is transmitted as the beamforming feedback, namely, $\mathbf{\Phi}$ in (\ref{eq:phi_bf}) and $\mathbf{\Psi}$ in (\ref{eq:psi_bf}). However, if only $\mathbf{\Phi}$ angle information is transmitted in the feedback, the number of angles that need to be reported in the feedback is halved. For $\mathbf{\Psi}$ angles, we may consider using a fixed vector of values that is known both to the AP as well as the STA. This approach was considered in the IEEE 802.11ah standard \cite{11ah} when only a single spatial stream is used in transmission. In this work, we extend this idea to multiple spatial streams. 

\begin{figure}[t] 
    \centering\vspace{-3mm}
    \includegraphics[width=0.9\linewidth]{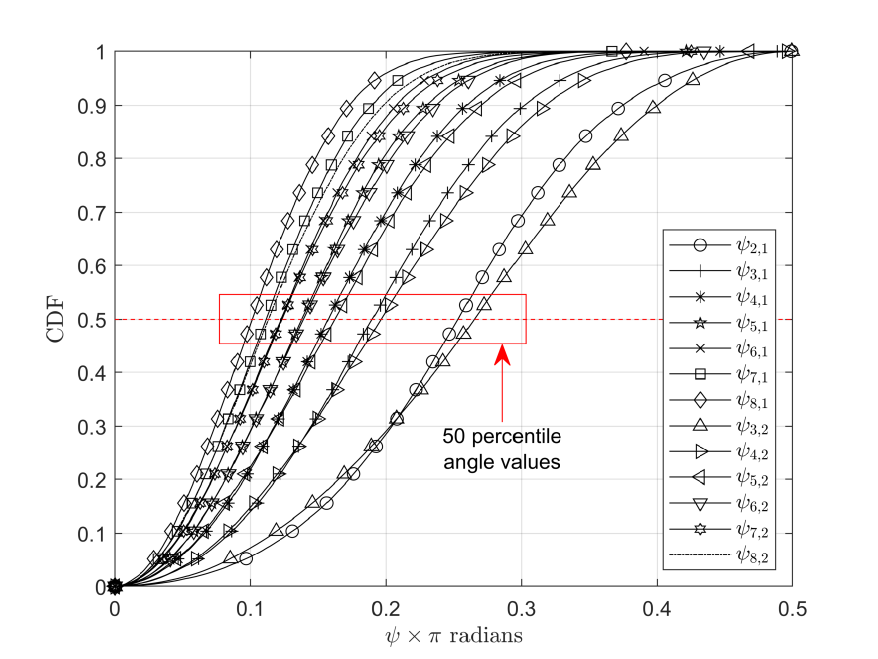}\vspace{-4mm}
    \caption{CDF of the $\psi$ angle data collected over $10^6$ packets in an $8\times 2$ MIMO case.}
    \label{fig:psi_cdf}\vspace{-4mm}
\end{figure}

In this case, the fixed angle values for $\mathbf{\Psi}$ may be determined using empirical data. To obtain these values, we run the downlink link-level simulation using compressed beamforming for a large number of packets (e.g., $10^6$) for a given MIMO configuration and for given values of $N_\text{c}$  and $N_\text{r}$. We store the raw angle values (before converting to the corresponding quantized index) of all the $\mathbf{\Psi}$ angles for all the subcarriers of all the simulated packets. From this generated data, we obtain the $50$th percentile values of all the $\mathbf{\Psi}$ angles and use them as our fixed values instead of computing them for the compressed beamforming representation as shown in (\ref{eq:psi_bf}). The beamformee and the beamformer both will have complete information about the fixed vector of the $\mathbf{\Psi}$ angle values. As an example, Fig.~\ref{fig:psi_cdf} shows the CDF of the $\mathbf{\Psi}$ angle data collected over $10^6$ packets and across $242$ subcarriers for a $8\times 2$ MIMO case. For $8\times 2$ MIMO, there are thirteen $\mathbf{\Psi}$ angles in the compressed beamforming representation. 

While this approach may result in loss of accuracy in generating $\mathbf{V}$ matrix at the beamformer, the reduction in the required number of bits in the beamforming feedback may be significant enough to result in a gain in the link goodput, as will be shown later. We further demonstrate this using the following example, where we show a lower bound for the GCS using the fixed $\psi$ values. The GCS is an intermediate key performance indicator (KPI) to measure the accuracy of the beamforming feedback.

\subsection{Partial Compressed Beamforming Feedback Example: }\label{sec:21_gcs_eg}

Consider a case where $N_\text{c}=1$ and $N_\text{r}=2$. In this case, the compressed beamforming representation of $\mathbf{V}$ reduces to 
\begin{equation}
    \mathbf{V} = \mathbf{D}_1 \mathbf{G}^T_{21}(\psi_{2,1})\times \Tilde{ \mathbf{I}}_{2\times 1}. \label{eq:v21_rep}
\end{equation}
The matrix $\mathbf{D}_1$ can be expressed as 
\begin{equation}
    \mathbf{D}_1 = \begin{bmatrix} 
        e^{j\phi_{1,1}} & 0\\ 
        0 & 1
        \end{bmatrix},
\end{equation}
the Givens rotation matrix $\mathbf{G}_{21}(\psi)$ is 
\begin{equation}
    \mathbf{G}_{21}(\psi) = \begin{bmatrix} 
        \cos(\psi_{2,1}) & \sin(\psi_{2,1}) \\ 
        -\sin(\psi_{2,1}) & \cos(\psi_{2,1}) \\ 
        \end{bmatrix},
\end{equation}
and $\mathbf{\tilde{I}} = [1~0]^T$. Substituting these in (\ref{eq:v21_rep}), we get
\begin{equation}
    \mathbf{V} = \begin{bmatrix} 
        \cos(\psi_{2,1}) e^{j\phi_{1,1}} \\ 
        \sin(\psi_{2,1}) 
        \end{bmatrix}. \label{eq:21_compressed}
\end{equation}

\begin{figure}[t] 
    \centering\vspace{-3mm}
    \includegraphics[width=0.9\linewidth]{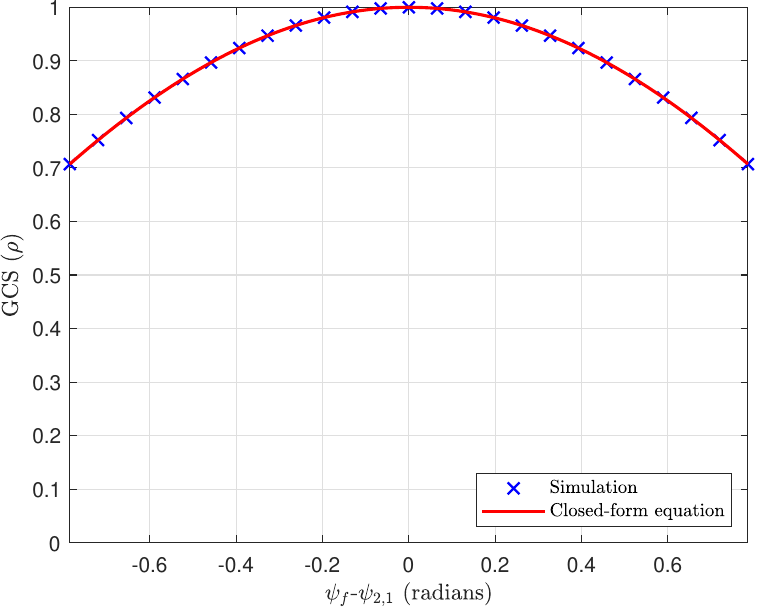}\vspace{-4mm}
    \caption{Variation in the GCS ($\rho$) for all the possible cases in partial compressed beamforming feedback for $N_\text{c}=1$ and $N_\text{r}=2$.}
    \label{fig:gcs_psi}\vspace{-4mm}
\end{figure}

\begin{table}[b]
  \centering
  \caption{Description of the simulation parameters.}
  \label{tab:sim_parameters}
  \begin{tabular}{|p{0.25\linewidth}|p{0.15\linewidth}|p{0.42\linewidth}|}
    \hline
    \textbf{Parameter} & \textbf{Value} & \textbf{Description} \\
    \hline
    $T_{\text{NDPA}}$ & $28~\mu s$ & Duration of NDPA frame \\
    \hline
    $T_{\text{NDP}}$ & $48+(N_r\times8)~\mu s$ & Duration of NDP frame \\
    \hline
    $T_{\text{SIFS}}$ & $16~\mu s$ & Short inter-frame spacing \\
    \hline
    $T_{\text{preamble}}$ & $64~\mu s$ & Duration of preamble \\
    \hline
    $T_{\text{ACK}}$ & $50~\mu s$ & ACK transmission duration \\
    \hline
    BW & $20$ MHz & Channel bandwidth \\
    \hline
    No. of Subcarriers & $242$ &  \\
    \hline
    Guard Interval & $0.8~\mu s$ & Guard interval for the data field \\
    \hline
    $N_\text{g}$ & $4$ & Subcarrier grouping \\
    \hline
    $b_{\phi}$ & $6$ & No. of bits used to quantize $\phi$ \\
    \hline
    $b_{\psi}$ & $4$ & No. of bits used to quantize $\psi$ \\
    \hline
    $P_0$ & $10^{-2}$ & Target PER for MCS selection \\
    \hline    
  \end{tabular}
\end{table}

From (\ref{eq:21_compressed}), it can be seen that the $\psi$ angles affect the amplitude of the elements in $\mathbf{V}$ and the $\phi$ angles affect the phase of $\mathbf{V}$. The matrix $\mathbf{V}$ is unitary, so it should satisfy $\mathbf{V}^H\mathbf{V}=1$. Now consider that we use a fixed value ($\psi_\text{f}$) for the $\psi$ angle. The new matrix $\mathbf{V}_\text{f}$ may then be expressed as 
\begin{equation}
	\mathbf{V}_\text{f}= \begin{bmatrix} 
        \cos(\psi_\text{f}) e^{j\phi_{1,1}} \\ 
        \sin(\psi_\text{f}) 
        \end{bmatrix}.
\end{equation}

The GCS to compare $\mathbf{V}_\text{f}$ to the real $\mathbf{V}$ can be computed as 
\begin{align}
    \rho = \lvert \mathbf{V}^H\mathbf{V}_\text{f} \rvert &= \Bigg\lvert \begin{bmatrix} 
        \cos(\psi_{2,1}) e^{j\phi_{1,1}} \\ 
        \sin(\psi_{2,1}) 
        \end{bmatrix}^H
	 \begin{bmatrix} 
        \cos(\psi_{f}) e^{j\phi_{1,1}} \\ 
        \sin(\psi_{f}) 
        \end{bmatrix} \Bigg\rvert \nonumber \\
	&= \lvert \cos(\psi_{2,1})\cos(\psi_\text{f}) + \sin(\psi_{2,1})\sin(\psi_\text{f}) \rvert \nonumber \\
	&= \lvert \cos(\psi_\text{f}-\psi_{2,1}) \rvert.
\end{align}

For the $2\times 1$ MIMO case, the fixed value chosen is $\psi_\text{f} = 0.25\pi$, and $\psi_{2,1}\in \{0, 0.5\pi\}$. Therefore, $(\psi_\text{f}-\psi_{2,1}) \in \{-0.25\pi, 0.25\pi\}$. Fig.~\ref{fig:gcs_psi} shows the GCS values for this range of angle values. It can be seen that there is a lower bound on the GCS equal to $\cos(0.25\pi) = \cos(-0.25\pi) = 1/\sqrt{2}$. With this lower bound, the performance loss due to loss in accuracy in $\psi$ angles is relatively small, as will be demonstrated in the results section.

\section{Simulation Results} \label{sec:sim_results}
In this section, we evaluate the methods described in this paper and compare them with two baseline techniques: IEEE 802.11be (compressed beamforming) and our previously proposed technique work iFOR, presented in \cite{9860553}.  

\begin{figure}[t] 
    \centering\vspace{-3mm}
    \includegraphics[width=0.9\linewidth]{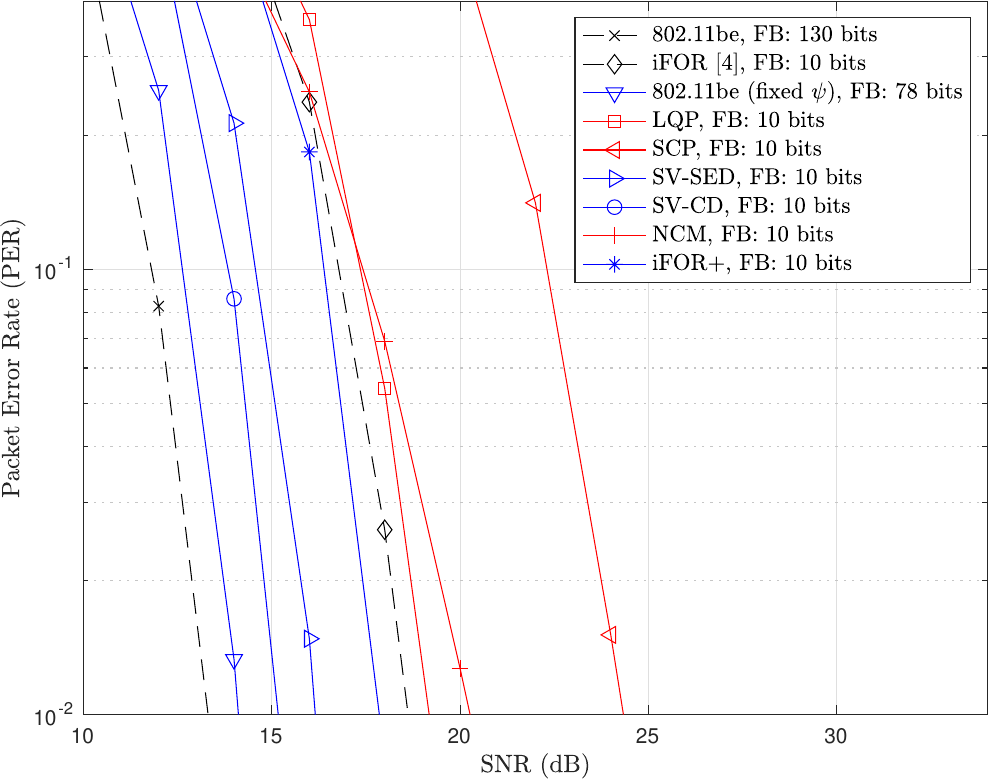}\vspace{-4mm}
    \caption{PER versus SNR comparison for $8\times 2$ MIMO with  MCS index 4.}
    \label{fig:per_mcs4}\vspace{-4mm}
\end{figure}

\subsection{Simulation Setup}
For all the subsequent results in this section, we consider a SU-MIMO downlink transmission complying with the IEEE 802.11be Wi-Fi standard. All the simulations are based on the toolboxes offered by MathWorks. Table~\ref{tab:sim_parameters} summarizes the common simulation parameters for the results presented in this section. For the \textit{LQP} method, we use two bits to quantize the $\psi$ angles instead of four. For our simulations, we consider channel models A-E defined by the IEEE 802.11 working group \cite{chan_model}. The key characteristics of these channel models are summarized in Table~\ref{tab:chan_models}, where $T_{\text{RMS}}$ is the RMS delay spread, $T_{\text{max}}$ is the maximum delay spread, and $N_{\text{taps}}$ is the number of taps in the channel. We use combined data generated from these channel models to generate our candidate sets using \textit{k}-means clustering. The PER results, however, are simulated using the channel model D. Henceforth, we refer to a scenario with $N_\text{r}$ antennas at the beamformer and $N_\text{c}$ number of spatial streams as the $N_\text{r}\times N_\text{c}$ MIMO case. Unless mentioned otherwise, we consider an $8\times 2$ SU-MIMO transmission in the downlink in all our simulations. The payload for each iteration of a simulation is $1,000$ bytes. For all the simulations related to the index-based methods, we use $1024$ candidates, which require $10$ bits of feedback per subcarrier group in the BFR. In the \textit{SCP} method, we use $256$ candidates for the $\phi$ angles, and $4$ candidates for the $\psi$ angles, and combine them to form $1024$ candidates.

\begin{figure}[t] 
    \centering\vspace{-3mm}
    \includegraphics[width=0.9\linewidth]{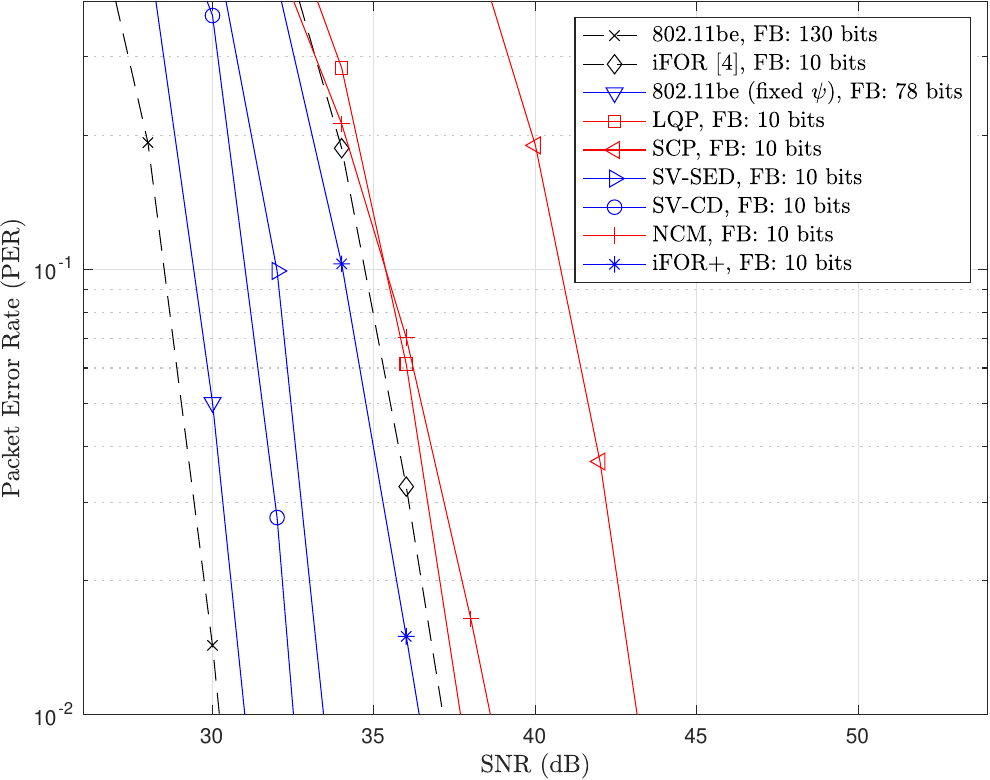}\vspace{-4mm}
    \caption{PER versus SNR comparison for $8\times 2$ MIMO with MCS index 11.}
    \label{fig:per_mcs11}\vspace{-4mm}
\end{figure}

\begin{table}[b]
  \centering
  \caption{Properties of the channel models.}
  \label{tab:chan_models}
  \begin{tabular}{|p{0.2\linewidth}|p{0.1\linewidth}|p{0.1\linewidth}|p{0.1\linewidth}|p{0.1\linewidth}|p{0.1\linewidth}|}
    \hline
    \multirow{2}{*}{\textbf{Parameter}} & \multicolumn{5}{c|}{\textbf{Channel Model}} \\
    \cline{2-6}
     & A & B & C & D & E \\
    \hline
    $T_{\text{RMS}}$ (ns) & 0 & 15 & 30 & 50 & 100 \\
    \hline
    $T_{\text{max}}$ (ns) & 0 & 80 & 200 & 390 & 730 \\
    \hline
    $N_{\text{taps}}$ & 1 & 9 & 14 & 18 & 18 \\
    \hline
  \end{tabular}
\end{table}

\subsection{Packet Error Rate Comparison}

Fig.~\ref{fig:per_mcs4} and Fig.~\ref{fig:per_mcs11} depict the PER results of the aforementioned methods for MCS $4$ and $11$, respectively. The fixed $\psi$ method offers a modest reduction in the feedback bits and performs well in terms of the PER, being the closest to the 802.11be baseline. Using fixed values for $\psi$ angles in the BFR enables transmitting information about only half of the angles compared to the traditional compressed beamforming method. As demonstrated with the example in Section~\ref{sec:21_gcs_eg}, using fixed $\psi$ values does not lead to a critical loss in accuracy of the $\mathbf{V}$ matrix, thus enabling decent PER performance. But even though the PER performance for this method is the best compared to the index-based methods we propose, it still requires a lot more information ($78$ bits per subcarrier group) compared to the $10$ bits per subcarrier group required for the index-based methods. This will have implications on the throughput gain for this method, as will be discussed later in this section. 

In \textit{LQP}, we use a lower quantization order for the $\psi$ angles, to reduce the impact of $\psi$ angle index values in the distance calculation while clustering and generation of candidates. However, similar to the case in the fixed $\psi$ method, losing accuracy in the $\psi$ angle results in a slight degradation in PER. This is again evident here noticing that the PER of \textit{LQP} is slightly worse than that of iFOR \cite{9860553}. In the \textit{SCP} method, we cluster the $\phi$ and the $\psi$ angles separately and then combine those candidates to generate our full candidate set. In compressed beamforming, even though the $\phi$ and the $\psi$ angles are calculated independently, their combined use is important to the representation of the $\mathbf{V}$ matrix. Thus, when clustering over them separately, the resulting candidates suffer from a loss in coherency and result in an additional loss in accuracy. This results in the candidates being less than ideal and resulting in the worst PER performance from all the methods proposed. 

Using \textit{SV-SED} and \textit{SV-CD} however, we are able to capture the pair-wise correlation between the candidates (or cluster centroids while clustering) and the real channel feedback vectors. This results in the best PER performance compared to all other index-based methods we have proposed. The \textit{SV-SED} method uses SED as the metric while clustering, which is not able to capture the pair-wise correlation between two vectors as well as the \textit{SV-CD} method does using cosine distance. We explain why using the following example. The serialized representation of a steering matrix has been shown in (\ref{eq:serialV}). Consider two such serialized vectors $\mathbf{v}_{\text{s,a}}$ and $\mathbf{v}_{\text{s,b}}$. The SED between these two vectors $\mathbf{v}_{\text{s,a}}$ and $\mathbf{v}_{\text{s,b}}$ can be calculated as 
\begin{align}
    d_{\text{SED}} &= \lvert \mathbf{v}_{\text{s,a}} - \mathbf{v}_{\text{s,b}} \rvert ^2 \nonumber\\
	&= \lvert \mathbf{v}_{\text{s,a}} \rvert ^2 + \lvert \mathbf{v}_{\text{s,b}} \rvert ^2 - \text{real}(\mathbf{v}_{\text{s,a}}^H\mathbf{v}_{\text{s,b}}). \label{eq:sv_sed}
\end{align}

\begin{figure}[t] 
    \centering\vspace{-3mm}
    \includegraphics[width=0.9\linewidth]{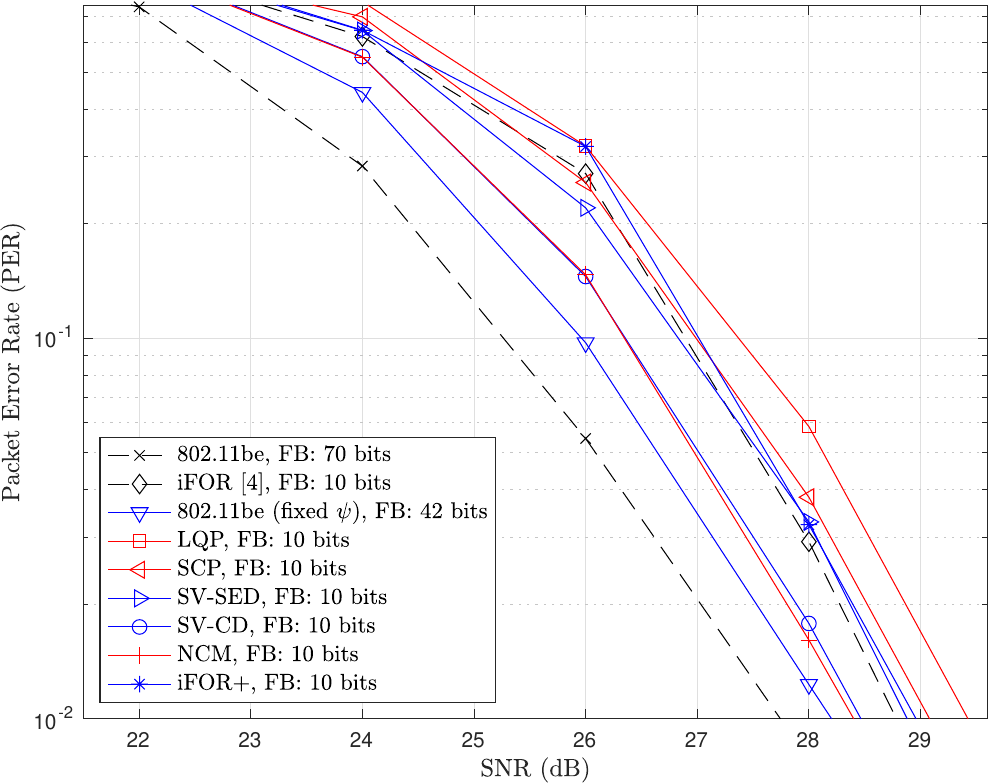}\vspace{-4mm}
    \caption{PER versus SNR comparison for $8\times 1$ MIMO with MCS index 11.}
    \label{fig:per_mcs11_8x1}\vspace{-4mm}
\end{figure}

From (\ref{eq:sv_sed}), it is clear that the SED between $\mathbf{v}_{\text{s,a}}$ and $\mathbf{v}_{\text{s,b}}$ is inversely proportional to the real part of the inner product between the two vectors. Whereas, in (\ref{eq:d_cd}) it has been shown that the CD between these two vectors is inversely proportional to the inner product. Hence, the \textit{SV-CD} method does a better job of encapsulating the element-wise correlation between two vectors resulting in better PER performance.

Using the \textit{iFOR+} method, we account for the effective distances between the $\phi$ angle indexes while generating the candidate set, as discussed in Section~\ref{sec:ifor_plus}. The corresponding PER results in Fig.~\ref{fig:per_mcs4} and Fig.~\ref{fig:per_mcs11} demonstrate that correcting for the effective distance enables the \textit{iFOR+} method to outperform the \textit{iFOR} \cite{9860553} results, albeit the performance is still not comparable to the 802.11be baseline.

For the $8\times 2$ MIMO case, using the \textit{NCM} method, the covariance matrices in (\ref{eq:cov_mat_1}) are normalized using the sum of squares of the singular values in the $\mathbf{S}$ matrix. Normalizing using this value does not result in the resultant matrix being unitary, which affects the clustering algorithm. The result is non-ideal candidates which leads to worse performance than that of \textit{iFOR}.

In Fig.~\ref{fig:per_mcs11_8x1}, we compare the PER results for the different proposed methods for an $8\times 1$ MIMO link with MCS index $11$. For the single spatial stream case, we demonstrate in (\ref{eq:app_2}) in the appendix that the SED between two \textit{NCMs} is inversely proportional to the inner product of the respective steering matrices. Hence, for the $8\times 1$ MIMO case, the PER performance of the \textit{NCM} method closely matches that of the \textit{SV-CD} method. The accuracy requirement for beamforming feedback for $8\times 1$ MIMO is not as high as for $8\times 2$ MIMO, resulting in the performance of all the index-based feedback methods being closer to each other than in the $8\times 2$ MIMO case. For the \textit{LQP} method,  since we lower the quantization order of the $\psi$ angles, there is a small degradation in the PER performance compared to \textit{iFOR}. 

\subsection{MCS Selection}\label{sec:mcs_discussion}

In this subsection, we discuss the MCS selection and the subsequent implications on the goodput performance. The simulation results displayed in Fig.~\ref{fig:mcs_selection} are for an $8\times 2$ MIMO case, where MCS index selection is done for a given SNR value based on the pre-determined PER threshold $P_0$ that the link should satisfy. For the results in Fig.~\ref{fig:mcs_selection}, $P_0$ is set to be $10^{-2}$, and the highest MCS index from the PER simulations that satisfies this threshold is selected. The selected MCS index determines the chosen transmission rate for transmitting the data payload, with a higher MCS index allowing for a higher data rate. Hence, better PER performance may enable choosing a higher MCS index, which may result in a higher chosen data rate and goodput. Since the 802.11be baseline has the best PER performance, using it leads to the highest selected MCS index in general. This is followed by the 802.11be baseline with fixed $\psi$, as the PER performance using this method is close to the baseline. On the other end, since using the \textit{SCP} method results in the worst PER performance amongst all for the $8\times 2$ MIMO case, it also leads to the selection of the lowest MCS index. 

\begin{figure}[t] 
    \centering\vspace{-3mm}
    \includegraphics[width=0.9\linewidth]{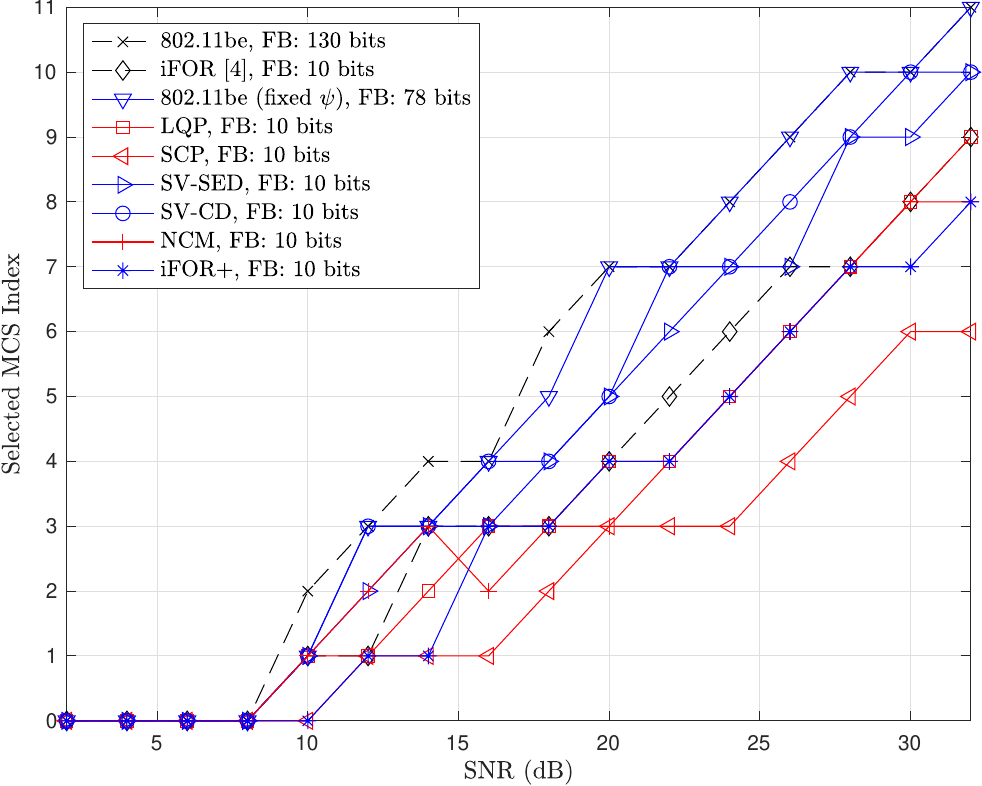}\vspace{-4mm}
    \caption{MCS index selection for PER threshold of $10^{-2}$.}
    \label{fig:mcs_selection}\vspace{-4mm}
\end{figure}

\subsection{Goodput Calculations}

Broadly, the goodput ($\Gamma$) of a link may be characterized as a measure of the effective or error-free data packet throughput. It may be considered as the ratio of the data payload in bits ($L_{\text{data}} $) transmitted in error-free packets and the total time ($T_{\text{total}}$) required to transmit this data.

Now, for a single Wi-Fi link, the total transmission duration is given by 
\begin{equation}
    T_{\text{total}} = T_{\text{sounding}} + T_{\text{data}} + T_{\text{SIFS}} + T_{\text{ACK}},
\end{equation}
where $T_{\text{sounding}}$ represents the  sounding duration, $T_{\text{data}}$  represents the data payload transmission duration, and $T_{\text{ACK}}$ is the time required to transmit the ACK message, as shown previously in Section~\ref{sec:sys_model}.

Considering the elements involved in channel sounding for an SU-MIMO case, the total sounding duration may be calculated as
\begin{equation}
     T_{\text{sounding}} = T_{\text{NDPA}} + T_{\text{SIFS}} + T_{\text{NDP}} + T_{\text{SIFS}} + T_{\text{BFR}}.
\end{equation}

For transmitting the BFR, the STA uses a MU-PPDU (see Fig.~\ref{fig:ppdu_format}) frame with the beamforming feedback information in the DATA field. Adhering to the same structure as the MU-PPDU, the transmission duration for the BFR can be computed as follows
\begin{equation}
    T_{\text{BFR}} = T_{\text{preamble}}+\frac{L_{\text{BFR}}}{R_{\text{BFR}}},
\end{equation}
where $L_{\text{BFR}}$  is the required number of bits to transmit the sounding feedback information and $R_{\text{BFR}}$  is the chosen data rate for transmitting this information. Similarly, for data transmission duration, 
\begin{equation}
    T_{\text{data}} = T_{\text{preamble}}+\frac{L_{\text{data}}}{R_{\text{data}}},
\label{eq:t_data}
\end{equation}
where $L_{\text{data}}$  is the required number of bits to transmit the payload data and $R_{\text{data}}$ is the chosen transmission data rate.

Now, coming back to the goodput calculations for the Wi-Fi link, the mathematical formula for goodput may be represented as follows  
\begin{equation}\label{eq:goodput_2}
    \Gamma = \frac{L_{\text{data}}}{T_{\text{sounding}} + T_{\text{data}}/(1-P_\text{e})  + T_{\text{SIFS}} + T_{\text{ACK}}},
\end{equation}
where $P_\text{e}$ represents the packet-error-rate. For the calculations in (\ref{eq:goodput_2}), we assume that the sounding procedure is not repeated in cases where data packet re-transmission is required. The values of the fixed parameters in the above calculations are listed in Table~\ref{tab:sim_parameters}.

\subsection{Goodput Comparison}

In this subsection, we discuss the goodput comparison of all the results presented in Fig.~\ref{fig:goodput}, where we plot the goodput derived in (\ref{eq:goodput_2}) against SNR. We discuss in Section~\ref{sec:mcs_discussion} that higher MCS index selections allow for higher data rates during the transmission of the data payload. Considering (\ref{eq:t_data}), we can say that a higher data rate would lead to lower data transmission duration leading to improvement in the goodput in (\ref{eq:goodput_2}). However, another critical parameter in the goodput calculations is $T_{\text{sounding}}$, which depends on the amount of information required in the BFR. We have seen that the 802.11be baseline has the best PER performance and the highest chosen data rate. However, the amount of bits ($130$ bits per subcarrier group) required in the BFR for the 802.11be baseline is high enough that the goodput gain is actually hampered and hence it has the worst goodput performance compared to all other methods we present, including iFOR. 

The 802.11be baseline with fixed $\psi$ chooses the data rates that are very close to the baseline, as can be inferred from Fig.~\ref{fig:mcs_selection}. The fixed $\psi$ method, however, offers a significant reduction in the number of bits required in the BFR (down to $78$ per subcarrier from $130$ in the baseline). Hence, overall, the fixed $\psi$ method provides a significant gain, of up to approximately $17\%$ in high SNR.  

For all the index-based methods, however, the number of bits required in the BFR is the same, since they all use $10$ bits of feedback per subcarrier group for these simulations. Moreover, since the target $P_\text{e}$ of our simulations is $10^{-2}$, minor differences in PER between the different index-based methods do not significantly affect the goodput in (\ref{eq:goodput_2}). The driving factor of the goodput results for the index-based methods is therefore the chosen data rate. Since the other index-based methods have a better PER performance than the \textit{SCP} method, they also exhibit better goodput performance in general. But because of a significant reduction in the BFR, the \textit{SCP} method still provides a high gain compared to the 802.11be baseline which is approximately $46\%$ at high SNR. It has been shown that the \textit{SV-CD} method provides the best PER results. Combined with the offered reduction in the bits of information required in the BFR, this method also provides the best goodput performance with a gain of approximately $54\%$ compared to the 802.11be baseline at high SNR.

\section{Conclusion and future work} \label{sec:conclusion}

In this paper, we proposed several machine learning-aided index-based methods for reducing beamforming feedback overhead in a Wi-Fi link and compare their performance to the IEEE 802.11be standard and our previous work in \cite{9860553}. 

\begin{figure}[t] 
    \centering\vspace{-3mm}
    \includegraphics[width=0.9\linewidth]{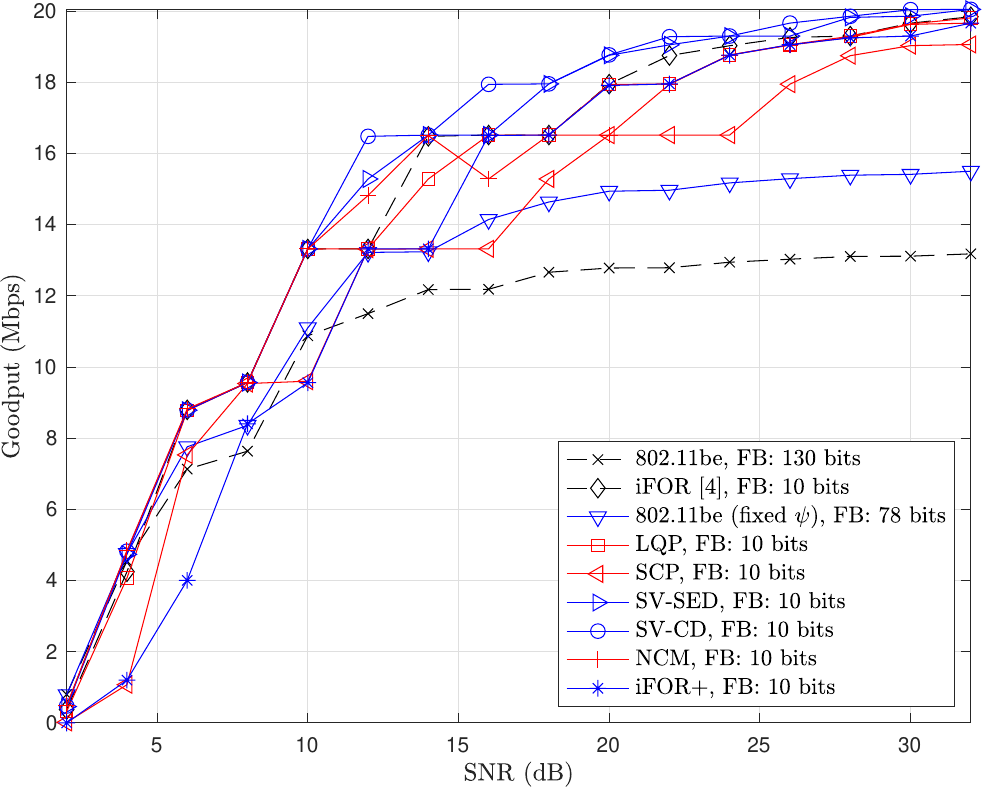}\vspace{-4mm}
    \caption{Goodput comparison for a payload of $1,000$ bytes.}
    \label{fig:goodput}\vspace{-4mm}
\end{figure}

Based on extensive simulation results, our analysis indicates that the representation of the data and the specific distance metric used for the candidate set generation have profound implications for the link performance. We show that, if the \textit{effective} distance between the $\phi$ angles is accounted for, there is an improvement in the PER performance. Whereas, if the quantization order for $\psi$ angles is changed or the $\phi$ and $\psi$ angles are clustered separately there is a degradation in the PER. Clustering over vectors containing the actual complex elements of the beamforming matrix, as in the \textit{SV-SED} and the \textit{SV-CD} methods, leads to the most efficient candidate sets that provide the best PER results amongst all the index-based methods. The \textit{SV-CD} method edges out \textit{SV-SED} due to the ability to capture the pairwise correlation between two vectors more accurately. Using candidates in the form of normalized covariance matrices leads to poor performance in the case of multiple spatial streams due to inefficient normalization, but matches the best PER performance (with that of \textit{SV-CD}) for the results using a single spatial stream.

Even though the proposed methods incur a loss in PER performance, the index-based low-overhead methods reduce the feedback overhead by a significant margin that enables improvement in the goodput performance of up to $54\%$ when compared to the goodput with the IEEE 802.11be baseline. This goodput improvement, coupled with the ease of implementation of the proposed methods makes them a viable alternative to the recent trend in literature of using NN-based approaches for CSI feedback. The relatively simpler implementation and significant performance gains offered by these index-based approaches are promising for future wireless networks, including Wi-Fi.

Moreover, we examine an additional method to transmit partial compressed feedback information in the BFR and the corresponding implications in the link performance. For this method, we use a simple $2\times 1$ MIMO example to demonstrate a lower bound on the GCS which indicates that the loss in accuracy incurred due to partial compressed beamforming information is not critical enough to hamper the PER, and in turn the goodput performance. We leave the detailed derivation for a generalized $M\times N$ MIMO case for our future work.

\bibliographystyle{IEEEtran}

\bibliography{ref}


\appendix[Proof of Lemma 1]
Here we derive the SED between two covariance matrices $\mathbf{K}_a$ and $\mathbf{K}_b$ and prove that it is inversely proportional to the inner product between $\mathbf{v}_a$ and $\mathbf{v}_b$. In the subsequent derivations, $\lVert.\rVert_\text{F}$ represents the Frobenius norm of the argument. For the case with a single spatial stream, $\lVert\mathbf{K}\rVert_\text{F}=\sqrt{(s^2)^2}=s^2$, since other matrices in (\ref{eq:norm_cov_mat}) are unitary. Now, the SED between $\mathbf{K}_a$ and $\mathbf{K}_b$ from (\ref{eq:d_cov_sed}) can be re-written as:
\begin{align}
    d_{\text{cov}} &= \sum^M_{i=1} \sum^M_{j=1} \Bigg| \frac{\mathbf{K}_a(i,j)}{\lVert\mathbf{K}_a\rVert_\text{F}} - \frac{\mathbf{K}_b(i,j)}{\lVert\mathbf{K}_b\rVert_\text{F}} \Bigg|^2  \label{eq:app_1} \\
	&= \frac{\lVert\mathbf{K}_a\rVert_\text{F}^2}{\lVert\mathbf{K}_a\rVert_\text{F}^2} + \frac{\lVert\mathbf{K}_b\rVert_\text{F}^2}{\lVert\mathbf{K}_b\rVert_\text{F}^2} - \sum^M_{i=1} \sum^M_{j=1} \Bigg( \frac{\mathbf{K}_a(i,j)^*\mathbf{K}_b(i,j)}{\lVert\mathbf{K}_a\rVert_\text{F}\lVert\mathbf{K}_b\rVert_\text{F}} \nonumber \\
	&\quad + \frac{\mathbf{K}_b(i,j)^*\mathbf{K}_a(i,j)}{\lVert\mathbf{K}_a\rVert_\text{F}\lVert\mathbf{K}_b\rVert_\text{F}} \Bigg) \nonumber \\
        &= 2 - \frac{1}{\lVert\mathbf{K}_a\rVert_\text{F}\lVert\mathbf{K}_b\rVert_\text{F}} \Bigg( \sum^M_{i=1} \sum^M_{j=1} \mathbf{K}_a(i,j)^*\mathbf{K}_b(i,j) \nonumber \\
	&\quad + \sum^M_{i=1} \sum^M_{j=1} \mathbf{K}_b(i,j)^*\mathbf{K}_a(i,j) \Bigg). \label{eq:app_2}
\end{align}

From (\ref{eq:norm_cov_mat}), $\mathbf{K}_a(i,j)=s_a^2v_{a,i}v_{a,j}^*$ and $\mathbf{K}_b(i,j)=s_b^2v_{b,i}v_{b,j}^*$ respectively. Substituting these in (\ref{eq:app_2}), we have
\begin{align}
    d_{\text{cov}} &= 2 - \frac{1}{\lVert\mathbf{K}_a\rVert_\text{F}\lVert\mathbf{K}_b\rVert_\text{F}} \Bigg( \sum^M_{i=1} \sum^M_{j=1} s_a^2v_{a,i}^*v_{a,j}s_b^2v_{b,i}v_{b,j}^* \nonumber \\
	&\quad + \sum^M_{i=1} \sum^M_{j=1} s_b^2v_{b,i}^*v_{b,j}s_a^2v_{a,i}v_{a,j}^* \Bigg) \nonumber \\
    &= 2 - \Bigg( \sum^M_{j=1}(v_{b,j}^*v_{a,j})\sum^M_{i=1}(v_{a,i}^*v_{b,i}) \nonumber \\
	&\quad + \sum^M_{i=1}(v_{b,i}^*v_{a,i})\sum^M_{j=1}(v_{a,j}^*v_{b,j}) \Bigg) \nonumber \\
    &= 2 - \Big( (\mathbf{v}^H_a\mathbf{v}_b)^H(\mathbf{v}^H_a\mathbf{v}_b) + (\mathbf{v}^H_a\mathbf{v}_b)^H(\mathbf{v}^H_a\mathbf{v}_b) \Big) \nonumber \\
	&= 2(1-|\mathbf{v}_a^H\mathbf{v}_b|^2). \label{eq:app_3}
\end{align}
\hfill $\qed$

\vfill

\end{document}